\documentclass[11pt]{article}
\usepackage{fullpage}
\usepackage{amsthm,amsmath,amsfonts,amssymb}
\usepackage{xcolor}
\usepackage{bbm}
\usepackage{enumerate}
\usepackage{hyperref}
\usepackage{tikz}
\usetikzlibrary{arrows.meta,snakes,decorations.pathreplacing,decorations,shadows}
\tikzset{
  vertex/.style={circle, draw, fill=white, inner sep=0pt, minimum size=5mm},
  dir/.style={-{Stealth[length=2.2mm]}, semithick},
  undir/.style={semithick}
}
\usepackage{natbib}
\usepackage{cleveref}

\newtheorem{theorem}{Theorem}[section]

\newtheorem{lemma}[theorem]{Lemma}

\newtheorem{question}{Question}
\newtheorem{corollary}[theorem]{Corollary}

\newcommand{\ignore}[1]{}

\newif\ifnotesw\noteswtrue
   {\ifnotesw\marginpar[\hfill\(\top\)]{\(\top\)}\fi}%
   {\ifnotesw\marginpar[\hfill\(\bot\)]{\(\bot\)}\fi}

\newcommand{\mnote}[1]%
    {\ifnotesw\marginpar%
        [{\scriptsize\begin{minipage}[t]{\marginparwidth}
        \raggedleft#1%
                        \end{minipage}}]%
        {\scriptsize\begin{minipage}[t]{\marginparwidth}
        \raggedright#1%
                        \end{minipage}}%
    \fi}

\newcommand{\bra}[1]{\langle #1 |}

\newcommand{\ket}[1]{| #1 \rangle}

\newcommand{\ketbra}[2]{\ket{#1}\bra{#2}}

\newcommand{\QQ}{\mathbb{Q}}
\newcommand{\ZZ}{\mathbb{Z}}
\newcommand{\RR}{\mathbb{R}}

\newcommand{\YY}{\mathcal{Y}}
\newcommand{\EE}{\mathcal{L}}

\DeclareMathOperator{\Tr}{Tr}
\DeclareMathOperator{\Gal}{Gal}
\DeclareMathOperator{\diag}{diag}
\DeclareMathOperator{\sgn}{sgn}

\newcommand{\norm}[1]{\lVert #1 \rVert}

\title{Pretty good quantum state transfer via transcendental edge weights}
\author{
Addison Ballif\footnote{Brigham Young University Idaho, Rexburg, ID, \emph{addisonballif@gmail.com}}, Mark Kempton\footnote{Brigham Young University, Provo, UT, \emph{mkempton@mathematics.byu.edu}}, James B.  Larsen\footnote{University of Michigan, Ann Arbor, MI, \emph{jblarsen@umich.edu}}, Kellon Sandall\footnote{Brigham Young University, Provo, UT, \emph{kellon08@gmail.com}},\\ Christino Tamon\footnote{Clarkson University, Potsdam, NY, \emph{tino@clarkson.edu}}, and Trevor Wai\footnote{Brigham Young University, Provo, UT, \emph{trevorwai2000@gmail.com}}
}
\date{}

\begin{document}
\maketitle
\begin{abstract}
    We prove that if we take a rooted product of a circulant graph with universal perfect state transfer with a path of fixed length whose end edge is weighted with a transcendental number, then there is pretty good state transfer between any pair of endpoints of these paths.  As a consequence, in a path with an even number of vertices with transcendental weights on the two edges incident to the endpoints, there is pretty good state transfer.  
\end{abstract}


\section{Introduction}

Considerable work in recent years has shown the importance of state transfer via a quantum walk on a graph.  As first proposed by Bose \cite{bose2003quantum}, a collection of spin-1/2 particles (qubits) are represented by the nodes of a graph, and couplings between particles are represented by the edges of that graph.  The strength of a coupling is given by a weight on the edges, and if desired, a vertex weight (diagonal entry of the adjacency matrix) can be added to give an energy potential induced by a magnetic field.  Such a system is described by a quantum walk on the graph, which is given by the unitary matrix $U(t) = e^{-itA}$ where $A$ is the (weighted) adjacency matrix of the graph.  See \cite{bose2003quantum, christandl2004perfect, godsil2011perfect, godsil2012state, Kay2011} and references therein.  The most commonly studied quantum state transfer phenomenon is \emph{perfect state transfer} (PST) between nodes of the graph.  This phenomenon is said to occur between nodes $u$ and $v$ if there is a time $\tau$ such that $|\bra{u}U(\tau)\ket{v}|=1$.  We likewise say there is \emph{pretty good state transfer} (PGST) between $u$ and $v$ if, for all $\epsilon>0$, there is a time $\tau$ such that $|\bra{u}U(\tau)\ket{v}|>1-\epsilon$.

One of the most natural frameworks in which to consider PST or PGST is in a qubit chain modeled by a graph that is a simple path. It has been known for many years now that if we allow ourselves to adjust weights on all vertices and edges of a path, then these weights can be tuned to produce PST between endpoints of an arbitrary length path \cite{christandl2004perfect}.  However, it has also been known for some time that if we restrict our attention to simple, unweighted paths, then PST cannot occur if the length of the path exceeds 3 (see \cite{godsil2012state}).  Subsequent work in \cite{kempton2016perfect} showed that if we allow only vertex weights (a potential arising from magnetic fields at the vertices) but not weights on edges, then we still cannot have PST on paths of length 4 or more.  However, if we only ask for PGST, then it was shown in \cite{kempton2017pretty} that a vertex weight whose value is a transcendental number placed on the two endpoints of the path can induce PGST in any length of path.   All this leads us to investigate exactly how many edges or vertices must be weighted to be able to achieve PST or PGST.  In other words, how simple can we make our weighting to achieve the desired quantum state transfer phenomena?  This is the motivating question behind the present work.  

Considerable further research has been done that addresses when PGST can be achieved by adding transcendental weights to vertices (diagonal entries of the matrix).  In addition to paths mentioned above, work in \cite{kempton2017pretty} also addresses adding transcendental vertex weights to graphs with an involutional symmetry.  This is further generalized in \cite{eisenberg2019pretty} where PGST is obtained by putting a transcendental weight on vertices of asymmetric graphs, where there is a pair of vertices that are cospectral (a generalization of symmetery; see \cite{kempton2019characterizing}). The effect of combining edge and vertex weights in strongly regular graphs to achieve both PST and PGST is similarly studied in \cite{godsil2020state}. Finally, similar ideas are applied to a generalization of PGST called \emph{pretty good fractional revival} in the presence of transcendental vertex weights; this is studied in \cite{drazen2024pretty}.

We say that a graph exhibits \emph{universal perfect state transfer} if there is PST between any pair of vertices.  Universal PGST is defined in a similar way.  It was shown in \cite{Kay2011} that if we work with a real symmetric adjacency matrix, then if there is PST from $u$ to $v$ and from $u$ to $w$, then $v=w$.  Thus we cannot have perfect state transfer between non-disjoint pairs of vertices and as such, no real-weighted undirected graph with more than two vertices admits universal perfect state transfer.  However, if we allow complex weights with a Hermitian adjacency matrix, then universal PST becomes possible.  See \cite{cfghst14,cgkst17} for the study of universal PST and universal PGST.  In particular, \cite{cgkst17} gives constructions of circulant graphs that achieve universal PST. More recent work in \cite{song2026circulant} characterizes when PST can happen in oriented circulant graphs.  The work in \cite{aceghtwz2025chiral} generalizes the work in \cite{kempton2017pretty} by studying rooted products of a graph with universal PST with paths that have a transcendental weight on the end vertex.  That is, given a circulant graph $X$ with universal PST, attach a path of a fixed length $k$ to each vertex of $X$, and then put a transcendental weight on the end vertex.  It is shown in \cite{aceghtwz2025chiral} that this yields PGST between any pair of these end vertices.  Our main result is a variation on this idea.  

It is thus a natural question to ask how much it matters that the transcendental weight be placed on a vertex rather than an edge.  Motivated by this question, we consider the same kinds of rooted products that were studied in \cite{aceghtwz2025chiral}, but instead of putting the weight on a loop at the end vertex, we put a transcendental weight on the edge incident to the end vertex, and leave all vertices unweighted (so all diagonal entries of the adjacency matrix are 0).  The main result of this paper is the following.

\begin{theorem}\label{thm1}
    Construct a weighted graph by starting with a (complex-weighted) circulant graph $X$, attach a path of length $k$ to each vertex of $X$, and weight the edge incident to the endpoint of each of these paths with some transcendental number.  Then there is PGST between any pair of end vertices. 
\end{theorem}
An immediate corollary to this result is the following.

\begin{corollary}\label{cor1}
    Given a path with an even number of vertices, if we put a transcendental weight on the edge incident to the end vertex, then there is PGST between the two endpoints.  
\end{corollary}
Interestingly, our proof technique does not help us understand odd-length paths.  Thus for even-length paths, weighting the end edges has largely the same effect as putting a weight on the end vertices, but for odd-length paths, this may not be the case.  

The remainder of the paper is organized as follows.  In Section  \ref{sec:prelim} we will give the necessary preliminaries regarding PGST, universal state transfer, and our rooted product constructions and their characteristic polynomials.  We prove the main result, Theorem \ref{thm1} and Corollary \ref{cor1} in Section 3.  We end with some discussion in Section \ref{sec:concl}.  Some appendices with technical details around Galois theory and orthogonal polynomials are included at the end for reference.  

\section{Preliminaries}\label{sec:prelim}

\subsection{State Transfer}

Let $X$ be a graph, and let $A=A(X)$ be its adjacency matrix. Let $A=\sum_\lambda \lambda E_\lambda$ be the spectral decomposition of the adjacency matrix $A$.  

We say there is \emph{perfect state transfer} from vertex $a$ to vertex $b$ in the graph $X$ if there is some time $\tau>0$ with $|\bra{b} e^{-iA(X)\tau}\ket{a}|=1$. We say there is {\em pretty good state transfer} (PGST) from $a$ to $b$ in $X$ if
for any $\varepsilon > 0$, there is a time $\tau$ so that
$1 - \varepsilon \le |\bra{b} e^{-iA(X)\tau} \ket{a}|$.

Let $u, v \in V(X)$ be two vertices, and let $A_x$ denote the matrix obtained by deleting row and column $x$ from the matrix $A$.  We say $u$ and $v$  are cospectral if $\phi(A_u) = \phi(A_v)$ where $\phi(M)$ denotes the characteristic polynomial of the matrix $M$.  We say that $u$ and $v$ are \emph{strongly cospectral} if $E_\lambda e_u =\pm E_\lambda e_v$ for all $\lambda$.  






The following theorem is useful to determine whether pretty good state transfer occurs.

\begin{theorem} (Kronecker, see \cite{lz})\label{thm:kronecker}
Let $\lambda_1,\ldots,\lambda_d$ and $q_1,\ldots,q_d$ be arbitrary real numbers.
For an arbitrarily small $\epsilon$, the system of inequalities
\begin{equation}
	|\lambda_r\tau - q_r| < \epsilon \pmod{2\pi}, 
	\ \ \
	r=1,\ldots,d
\end{equation}
admits a solution for $\tau$ if and only if for integers $\ell_1,\ldots,\ell_d$ with
\begin{equation}
	\ell_1\lambda_1 + \ldots + \ell_d\lambda_d = 0
\end{equation}
implies
\begin{equation}
	\ell_1 q_1 + \ldots + \ell_d q_d \equiv 0\pmod{2\pi}.
\end{equation}
\end{theorem}



\subsection{Multiple State Transfer}

We say that a graph $X$ admits \emph{universal perfect state transfer} if it has perfect state transfer between every pair of vertices.  As mentioned in the introduction,  if we work with a real symmetric adjacency matrix, then we cannot have perfect state transfer between non-disjoint pairs of vertices and as such, no real-weighted undirected graph with more than two vertices admits universal perfect state transfer \cite{Kay2011}.  In contrast, there are many examples of complex-weighted Hermitian adjacency matrices which admit universal perfect state transfer. The study of universal state transfer was initiated in \cite{cfghst14} and \cite{cgkst17}, and we will use many tools developed in these papers.  

Let $X$ be a graph whose adjacency matrix has the following spectral decomposition:
\[
	A(X) = \sum_\lambda \lambda E_\lambda.
\]
Two vertices $a,b$ are {\em strongly cospectral} if there exist $q_\lambda(a,b) \in (-\pi,\pi]$ 
(or {\em quarrel}, see \cite{gl20}) for each eigenvalue $\lambda$ of $X$, so that
\begin{equation} \label{eqn:quarrel}
	E_\lambda \ket{a} = e^{iq_\lambda(a,b)}E_\lambda \ket{b}.
\end{equation}

Similarly, we say there is {\em universal pretty good state transfer} if for all pairs of vertices $a,b$ in $X$ and
for any $\varepsilon > 0$, there is a time $\tau$ so that
$1 - \varepsilon \le |\bra{b} e^{-iA(X)\tau} \ket{a}|$.

Therefore,
\begin{eqnarray}
1 - \varepsilon
	& \le & \left| \sum_\lambda e^{-i\lambda \tau} \bra{b}E_\lambda\ket{a} \right|,
		\ \ \mbox{ by spectral decomposition } \\
	& = & \left| \sum_\lambda e^{i(q_\lambda(a,b) - \lambda \tau)} \bra{b}E_\lambda\ket{b} \right|,
		\ \ \mbox{ by strong cospectrality (quarrel) } \\
	& \le & \sum_\lambda |\bra{b}E_\lambda\ket{b}|,
		\ \ \mbox{ by triangle inequality } \\
	& = & \sum_\lambda \bra{b}E_\lambda\ket{b},
		\ \ \mbox{ because $\bra{b}E_\lambda\ket{b}$ is nonnegative } \\
	& = & 1,
		\ \ \mbox{ since $\sum_\lambda E_\lambda = I$ }
\end{eqnarray}
Thus, we can see that if for each eigenvalue $\lambda$ we have
\begin{equation}
	|\lambda\tau - q_\lambda(a,b)| < \epsilon \pmod{2\pi}
\end{equation}
then there is PGST between $a$ and $b$.

The Discrete Fourier Transform (DFT) matrix of order $n$ is defined as
\begin{equation}
	F_n = \frac{1}{\sqrt{n}}
	\begin{pmatrix}
	1 & 1 & 1 & \ldots & 1 \\
	1 & \zeta & \zeta^2 & \ldots & \zeta^{n-1} \\
	1 & \zeta^2 & \zeta^4 & \ldots & \zeta^{2(n-1)} \\
	\vdots & \vdots & \vdots & \vdots & \vdots \\
	1 & \zeta^{n-1} & \zeta^{2(n-1)} & \ldots & \zeta^{(n-1)^2} 
	\end{pmatrix}
\end{equation}
where $\zeta = e^{2\pi i/n}$. 
Any circulant matrix of order $n$ is diagonalized by $F_n$.
Suppose $C$ is a circulant matrix where $F_n^\dagger C F_n = \diag(\lambda_0,\lambda_1,\ldots,\lambda_{n-1})$. 
 Then, the spectral decomposition for $C$ can be expressed  in terms of the columns of $F_n$.  If $\mathbf{z_k} = \begin{pmatrix}1&\zeta^k&\zeta^{2k}&\dots&\zeta^{k(n-1)}\end{pmatrix}^T$ then the spectral idempotents are given by $E_{\lambda_k}=\mathbf{z_k}\mathbf{z_k}^T$ in the spectral decomposition \[C=\sum_{k=0}^{n-1}\lambda_{k}E_{\lambda_k}.\]  As such, for a graph with adjacency matrix $C$, we can determine the quarrels explicitly: if the vertices of the circulant graph are labeled $0,1,\dots,n-1$, then for vertices $a$ and $b$, \begin{equation}\label{eq:quarrel}
	q_{\lambda_k}(a,b) = \frac{2\pi(b-a)k}{n},
	\ \ \mbox{ $k=0,1,\ldots,n-1$.}
\end{equation}

\subsection{Rooted Products}

It will be convenient to describe our constructions as \emph{rooted products} of graphs.  Let $X$ be a graph on $n$ vertices and let $\YY = \{Y_1(a_1),\ldots,Y_n(a_n)\}$ be a collection 
of rooted graphs where $a_i \in V(Y_i)$ is a special designated root vertex of $Y_i$. 
The {\em rooted product} $Z = X \circ \YY$ is the graph obtained by attaching $Y_i$ 
(through its root vertex $a_i$) to the $i$-th vertex of $X$.

\begin{theorem} (Godsil and McKay \cite{godsil1978new}) \label{thm:charpoly-rooted}
The characteristic polynomial of $Z$ is given by
\begin{equation}
	\phi(Z,t) = \det(D(t) - D_0(t)A(X))
\end{equation}
where $D(t) = \diag(\{\phi(Y_i,t)\}_{i=1}^{n})$ and 
$D_0(t) = \diag(\{\phi(Y_i \setminus a_i)\}_{i=1}^{n})$.
\end{theorem}

\subsection{Orthogonal Polynomials} 

Consider the Jacobi matrix given by
\begin{equation} \label{eqn:jacobi-matrix}
	\begin{pmatrix}
	a_m    & b_m        & 0                  & \ldots   & 0      & 0         & 0 \\
	b_m      & a_{m-1}  & b_{m-1}            & \ldots   & 0      & 0         & 0 \\
	0      & b_{m-1}    & a_{m-2}            & \ldots   & 0      & 0         & 0 \\
	\vdots & \vdots     & \vdots             &          & \vdots & \vdots    & \vdots \\
	0      & 0          & 0                  & \ldots   & a_3    & b_3       & 0 \\
	0      & 0          & 0                  & \ldots   & b_3    & a_2       & b_2 \\
	0      & 0          & 0                  & \ldots   & 0      & b_2       & a_1
	\end{pmatrix}
\end{equation}
By Favard's theorem (see Theorem \ref{thm:favard} in Appendix \ref{sec:app_orthpolyn}), we can construct a sequence of orthogonal polynomials 
$(Q_k(x))_{k=0}^{m-1}$ using the three-term recurrence induced by \eqref{eqn:jacobi-matrix}:
\begin{equation}
b_{k+1} Q_{k}(x) = (x - a_k)Q_{k-1}(x) - b_k Q_{k-2}(x), 
	\ \ \mbox{ $k=1,\ldots,m$}
\end{equation}
where $Q_0(x) = 1$, $Q_{-1}(x) = 0$ and $b_{m+1} = b_1 = 1$ (for convenience).

This can be expressed in the following convenient matrix equation.
If $x$ is a root of $Q_m(x)$ then
\begin{equation}
	\begin{pmatrix}
	a_m    & b_m        & 0                  & \ldots   & 0      & 0         & 0 \\
	b_m    & a_{m-1}    & b_{m-1}            & \ldots   & 0      & 0         & 0 \\
	0      & b_{m-1}    & a_{m-2}            & \ldots   & 0      & 0         & 0 \\
	\vdots & \vdots     & \vdots             &          & \vdots & \vdots    & \vdots \\
	0      & 0          & 0                  & \ldots   & a_3    & b_3       & 0 \\
	0      & 0          & 0                  & \ldots   & b_3    & a_2       & b_2 \\
	0      & 0          & 0                  & \ldots   & 0      & b_2       & a_1
	\end{pmatrix}
	\begin{pmatrix}
	Q_{m-1}(x) \\ 
	Q_{m-2}(x) \\
	Q_{m-3}(x) \\ 
	\vdots \\
	Q_2(x) \\
	Q_1(x) \\
	Q_0(x)
	\end{pmatrix}
	=
	x
	\begin{pmatrix}
	Q_{m-1}(x) \\ 
	Q_{m-2}(x) \\
	Q_{m-3}(x) \\ 
	\vdots \\
	Q_2(x) \\
	Q_1(x) \\
	Q_0(x)
	\end{pmatrix}
\end{equation}
As we will see below, our focus is on the special case where 
$a_1 = \ldots = a_{m-1} = 0$, $a_m = \lambda$, $b_2=\alpha$, and $b_3 = \ldots = b_m = 1$.
\ignore{
By Theorem \ref{thm:favard}, 
\begin{equation} \label{eqn:pseudo-chebyshev}
Q_0(x),Q_1(x),\ldots,Q_m(x)
\end{equation} 
forms a sequence of orthogonal polynomials with respect to a unique moment functional $\EE$.
That is, $Q_k(x)$ is a polynomial of degree $k$, $\EE[Q_j(x)Q_k(x)] = 0$ for $j \neq k$,
and $\EE[Q_k(x)^2] \neq 0$, for $k=0,1,\ldots,m$.
}

\section{Adding weighted paths to circulant graphs}
In this section, we will prove our main result that adding paths to each node of a circulant graph with universal PST and then adding a transcendental weight to the end edge of each path will yield PGST between the endpoints.  (A (weighted) graph is called a circulant graph if its adjacency matrix is circulant.)







Let $X$ be a circulant graph with universal PST. We focus on the rooted product with $Y_i(a_i) = P_m^\alpha(1)$, for each $i$, which
we will denote as $X \circ P_m^\alpha(1)$.
Recall vertex $1$ of $P_m^\alpha$ is the pendant vertex not adjacent to the $\alpha$-weighted edge.
The adjacency matrix of $Z$ can be written as an $m \times m$ block matrix:
\begin{equation} \label{eqn:rooted-matrix}
	A(Z) = 
	\begin{pmatrix}
	A(X) & I_n & 0 & \ldots & 0 & 0 \\
	I_n  & 0   & I_n & \ldots & 0 & 0 \\
	0    & I_n & 0 & \ldots & 0 & 0 \\
	\vdots & \vdots & \vdots & \vdots & \vdots \\
	0    & 0   & 0  &       & 0 & \alpha I_n \\
	0    & 0   & 0  &       & \alpha I_n & 0
	\end{pmatrix}
\end{equation}
We may index the vertices of $Z$ as follows: 
for $j=1,\ldots,m$ and $a \in V(X)$, let $(j,a)$ denote the $j$-th vertex on the path $P_m^\alpha$
that is rooted at vertex $a$ of $X$. So, $(m,a)$ denotes the vertex at the end (incident to the edge with the transcendental weight $\alpha$)  of the path  
rooted at vertex $a$ of $X$ and $(1,a)$ denotes the base of this path (which is identified with
the original vertex $a$ in $X$).


We now return to the adjacency matrix of the rooted product $Z$.
Suppose $\begin{pmatrix} y_m & y_{m-1} & \ldots & y_1 \end{pmatrix}^T$ is
an eigenvector of $Z$ with eigenvalue $\theta$:
\begin{equation} \label{eqn:jacobi}
	\begin{pmatrix}
	A(X) & I_n & 0 & \ldots & 0 & 0 \\
	I_n & 0 & I_n & \ldots & 0 & 0 \\
	0 & I_n & 0 & \ldots & 0 & 0 \\
	\vdots & \vdots & \vdots & \vdots & \vdots \\
	0 & 0 & 0  &       & 0 & \alpha I_n \\
	0 & 0 & 0  &       & \alpha I_n & 0
	\end{pmatrix}
	\begin{pmatrix}
	y_{m} \\ 
	y_{m-1} \\
	y_{m-2} \\
	\vdots \\
	y_2 \\
	y_1
	\end{pmatrix}
	=
	\theta
	\begin{pmatrix}
	y_{m} \\ 
	y_{m-1} \\
	y_{m-2} \\
	\vdots \\
	y_2 \\
	y_1
	\end{pmatrix}
\end{equation}
It is clear that each entry $y_i$ may depend on $\theta$.
Therefore, from the last $m-1$ equations, we obtain
\begin{eqnarray}
y_1 & = & Q_0(\theta) y_1,
	\ \ \mbox{ since $Q_0(x) = 1$ } \\
y_2 & = & (\theta/\alpha) y_1 = Q_1(\theta) y_1,
	\ \ \mbox{ since $Q_1(x) = x/\alpha$ } \\
y_3 & = & \theta y_2 - \alpha y_1 = (\theta Q_1(\theta) - \alpha Q_0(\theta)) y_1 = Q_2(\theta) y_1,
	\ \ \mbox{ since $Q_2(x) = xQ_1(x) - \alpha Q_0(x)$ } \\
\vdots & = & \vdots \\
\label{eqn:boo-1}
y_{m-1} & = & \theta y_{m-2} - y_{m-3} = (\theta Q_{m-3}(\theta) - Q_{m-4}(\theta)) y_1 = Q_{m-2}(\theta) y_1 \\
\label{eqn:boo-0}
y_m & = & \theta y_{m-1} - y_{m-2} = (\theta Q_{m-2}(\theta) - Q_{m-3}(\theta)) y_1 = Q_{m-1}(\theta) y_1
\end{eqnarray}
where the equalities follow from the three-term recurrence defining $Q_n(x)$.
On the other hand, the first equation of \eqref{eqn:jacobi} yields 
\begin{equation}
A(X) y_m + y_{m-1} = \theta y_m.
\end{equation} 

Now, suppose $y_1$ is an eigenvector of $A(X)$ with eigenvalue $\lambda$, that is:
\[
	 A(X)y_1 = \lambda y_1.
\]
Applying $y_{m} = Q_{m-1}(\theta) y_1$ from \eqref{eqn:boo-0} and $y_{m-1} = Q_{m-2}(\theta) y_1$ from \eqref{eqn:boo-1},
we have
\begin{equation}
(Q_{m-1}(\theta) A(X) + Q_{m-2}(\theta)I_n) y_1 = \theta Q_{m-1}(\theta) y_1
\end{equation}
which yields
\begin{equation}
0 = (\theta - \lambda) Q_{m-1}(\theta) - Q_{m-2}(\theta) = Q_m(\theta)
\end{equation}
where the last equality follows from the recurrence defining $Q_m(x)$.
So, $\theta$ is a root of $Q_m(x)$.

\medskip
\par\noindent{\em Remark}: From the above discussion, for a given eigenvalue $\lambda$ of $X$, 
we obtain a sequence of orthogonal polynomials which we will denote $(Q_k^\lambda(x))_{k=0}^{m-1}$ 
(to emphasize explicitly its dependence on $\lambda$).

\subsection{Eigenprojectors}

Let $x_\lambda$ be an eigenvector of $X$ with eigenvalue $\lambda$ and let $\theta$ be a root of $Q_m^\lambda(x)$.
Then, $y_\theta \otimes x_\lambda$ is an eigenvector of $Z$ with eigenvalue $\theta$, where
\begin{equation} \label{eqn:y-theta}
	y_{\theta} =
	\begin{pmatrix}
	Q_{m-1}^\lambda(\theta) & Q_{m-2}^\lambda(\theta) & \ldots & Q_0^\lambda(\theta)
	\end{pmatrix}^T.
\end{equation}
For two distinct roots $\theta_1$ and $\theta_2$ of $Q_m^\lambda(x)$, the eigenvectors $y_{\theta_1}$ and $y_{\theta_2}$
are orthogonal by the Christoffel-Darboux theorem (see Theorem \ref{thm:christoffel-darboux}).

\begin{lemma}
If $X$ has distinct eigenvalues, then the eigenvalues of $Z = X \circ P_m^\alpha(1)$ are given by
the eigenvalues of $Q_m^\lambda(x)$, as $\lambda$ ranges over the eigenvalues of $X$.
Moreover, the orthogonal projector corresponding to the eigenvalue $\theta$ of $Q_m^\lambda(x)$ 
is given by
\begin{equation}
	\Xi_\theta = M_\theta \otimes E_\lambda
\end{equation}
where $M_\theta = y_\theta y_\theta^\dagger/\norm{y_\theta}^2$ with $y_\theta$ defined by \eqref{eqn:y-theta},
and $E_\lambda = \ketbra{\lambda}{\lambda}$ is the orthogonal projector corresponding to the $\lambda$-eigenspace of $X$.
\end{lemma}

\begin{lemma} \label{lemma:shared-quarrels} (Shared Quarrels) \\
Let $X$ be an arbitrary graph and let $Z = X \circ P_m^\alpha(1)$ be the rooted product of $X$
with a path $P_m^\alpha$ on $m$ vertices (with a self-loop at vertex $m$).
Let $a,b$ be vertices of $X$ so that for each eigenvalue $\lambda$ of $X$
\begin{equation} 
E_\lambda \ket{a} = e^{iq_\lambda(a,b)}E_\lambda \ket{b},
\end{equation}
where $q_\lambda(a,b) \in (-\pi,\pi]$.
Then, for each eigenvalue $\theta$ of $Z$ where $\lambda$ is the corresponding eigenvalue of $X$, 
then
\begin{equation}
	\Xi_\theta \ket{j,a} = e^{iq_\lambda(a,b)} \Xi_\theta \ket{j,b},
\end{equation}
for each $j=1,\ldots,m$.
That is, $q_\theta((j,a),(j,b)) = q_\lambda(a,b)$, for each $j=1,\ldots,m$.

\begin{proof}
For each eigenvalue $\theta$ of $Z$, we have
\begin{equation}
	\Xi_\theta \ket{j,a} 
		= M_\theta \ket{j} \otimes E_\lambda \ket{a}
		= e^{iq_\lambda(a,b)} M_\theta \ket{j} \otimes E_\lambda \ket{b}
		= e^{iq_\lambda(a,b)} \Xi_\theta \ket{j,b}.
\end{equation}
\end{proof}
\end{lemma}

\subsection{Multiple state transfer} 

The theorem below provides a family of graphs with multiple state transfer, i.e., pretty good state transfer between multiple nodes. An example graph from this family is given in Figure \ref{fig:circulant-rooted-p3}.
\begin{figure}
    \centering

\begin{tikzpicture}[scale=1.1]
  \node[vertex] (A) at (90:1.4) {};
  \node[vertex] (B) at (210:1.4) {};
  \node[vertex] (C) at (330:1.4) {};

  \node[vertex] (A1) at (90:2.6) {};
  \node[vertex] (A2) at (90:3.8) {};
  \node[vertex, label=center:$a$] (A3) at (90:5.0) {};

  \node[vertex] (B1) at (210:2.6) {};
  \node[vertex] (B2) at (210:3.8) {};
  \node[vertex, label=center:$b$] (B3) at (210:5.0) {};

  \node[vertex] (C1) at (330:2.6) {};
  \node[vertex] (C2) at (330:3.8) {};
  \node[vertex, label=center:$c$] (C3) at (330:5.0) {};

  \draw[dir] (A) to[bend left=18] node[midway, left] {$i$} (B);
  \draw[dir] (B) to[bend left=18] node[midway, above right] {$i$} (C);
  \draw[dir] (C) to[bend left=18] node[midway, right] {$i$} (A);

  \draw[dir] (B) to[bend left=18] node[midway, left] {$-i$} (A);
  \draw[dir] (C) to[bend left=18] node[midway, below right] {$-i$} (B);
  \draw[dir] (A) to[bend left=18] node[midway, right] {$-i$} (C);

  \draw[undir] (A) -- (A1);
  \draw[undir] (A1) -- (A2);
  \draw[undir] (A2) -- node[midway, right] {$\alpha$} (A3);

  \draw[undir] (B) -- (B1);
  \draw[undir] (B1) -- (B2);
  \draw[undir] (B2) -- node[midway, above left] {$\alpha$} (B3);

  \draw[undir] (C) -- (C1);
  \draw[undir] (C1) -- (C2);
  \draw[undir] (C2) -- node[midway, above right] {$\alpha$} (C3);
\end{tikzpicture}

    \caption{Example of a graph satisfying \ref{thm:main}. Note that $C$ must be directed if it has more than 2 vertices, as no such graph with a real symmetric adjacency matrix can have universal perfect state transfer \cite{Kay2011}. By \ref{thm:main}, this graph will have PGST between nodes $a$ and $b$, $b$ and $c$, and $c$ and $a$.}
    \label{fig:circulant-rooted-p3}
\end{figure}
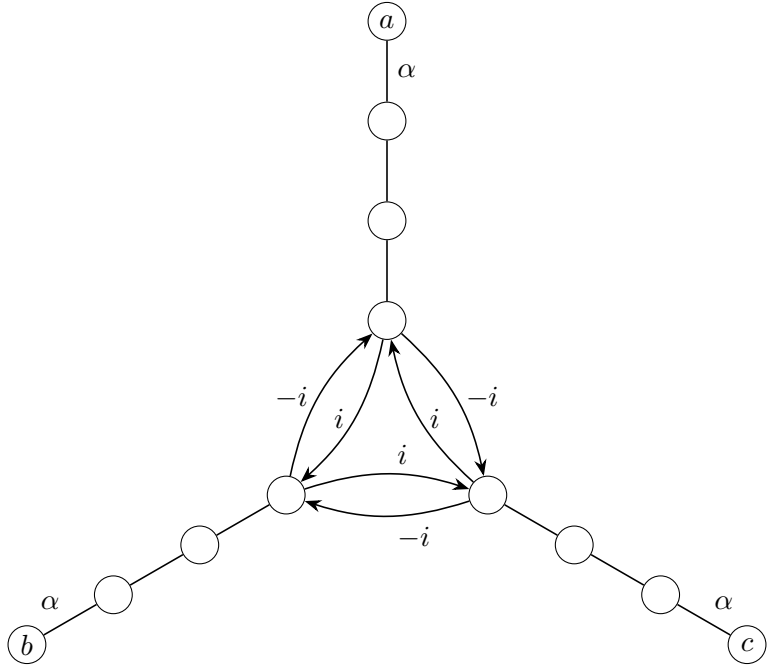

\begin{theorem} \label{thm:main}
Let $X$ be a circulant of order $n$ which has universal perfect state transfer and $\alpha$ be any transcendental real number.
Then, the rooted product $Z = X \circ P_m^\alpha(1)$ has pretty good state transfer between every pair of pendant 
vertices adjacent to the $\alpha$-weighted edges.

\begin{proof}
We use the machinery employed in \cite{kempton2017pretty} and \cite{eisenberg2019pretty}. 
We will prove pretty good state transfer from $(m,a)$ to $(m,b)$ in $Z$ for arbitrary vertices $a,b$ of $X$.
Suppose $A(Z) = \sum_\theta \theta \Xi_\theta$ is the spectral decomposition of $Z$.
We must show that for any $\varepsilon > 0$ there is $\tau$ so that
\begin{eqnarray}
1 - \varepsilon
	& \le & |\bra{m,b} e^{-iA(Z)\tau}\ket{m,a}| \\
	& = & |\sum_\theta e^{-i\theta\tau} \bra{m,b} \Xi_\theta \ket{m,a}| \\
	& = & |\sum_\theta e^{i(q_\theta((m,a),(m,b))-\theta\tau)} \bra{m,b} \Xi_\theta \ket{m,b}|.
\end{eqnarray}
Thus, we need to show there is $\tau$ so that
\begin{equation} \label{eqn:rooted-quarrel}
	|\theta\tau - q_\theta((m,a),(m,b))| < \epsilon \pmod{2\pi}
\end{equation}
for each eigenvalue $\theta$ of $Z$.
By Lemma \ref{lemma:shared-quarrels}, the quarrels of $Z$ are inherited from the quarrels of $C$ 
because $q_\theta((m,a),(m,b)) = q_\lambda(a,b)$ where 
$\lambda$ is an eigenvalue of $C$.
Let $C$ be the (possibly complex-weighted) adjacency matrix of $X$ and note that $C$ is a Hermitian circulant matrix. 

\bigskip

By Theorem \ref{thm:charpoly-rooted}, the characteristic polynomial of $Z$ is given by 
\begin{eqnarray}
\phi(Z,t) 
	& = & \det( \phi(P_m^\alpha,t) I_n - \phi(P_{m-1}^\alpha,t) C) \\
	& = & \det( \phi(P_m^\alpha,t) I_n - \phi(P_{m-1}^\alpha,t) \Lambda)
\end{eqnarray}
where $\Lambda=\diag(\lambda_0,\lambda_1,\dots,\lambda_{n-1})$ is the diagonal matrix consisting of the eigenvalues of $C$.  Note that $\Lambda = F_n^\dagger CF_n$ where $F_n$ is the discrete Fourier transform matrix of order $n$.
Therefore, the characteristic polynomial for $Z$ is given by
\begin{equation}
	\phi(Z,t) = \prod_{k=0}^{n-1} 
		(\phi(P_m^\alpha,t) - \lambda_k\phi(P_{m-1}^\alpha,t)).
\end{equation}
So, $\phi(Z,t)$ factors into $n$ polynomials from $\QQ(\alpha)[t]$.
We will denote the above polynomials as
$R_{\lambda_k}(t) = \phi(P_m^\alpha,t) - \lambda_k\phi(P_{m-1}^\alpha,t)$ 
and analyze $\phi(Z,t) = \prod_{k} R_{\lambda_k}(t)$.

\bigskip
\par\noindent{\em Notation}:
Each eigenvalue $\theta$ of $A(Z)$ shares a quarrel with and eigenvalue $\lambda$ of $C$.  
Thus, we will denote the eigenvalues of $Z$ corresponding to eigenvalue $\lambda_k$ of $C$
as $\theta_{k,j}$, where $j=1,\ldots,m$.  Note that $\theta_{k,j}$ for $j=1,\dots,m$ have quarrel $q_{\lambda_k}(a,b)$ and and are roots of the factor $R_{\lambda_k}(t)$ of the characteristic polynomial defined above.
\bigskip

By Lemma \ref{lemma:shared-quarrels}, for each $k$ we have
\[
	\Xi_{k,j}\ket{m,a} = e^{i q_{\lambda_k}(a,b)} \Xi_{k,j}\ket{m,b},
	\ \ \ j=1,\ldots,m.
\]
So, Equation \eqref{eqn:rooted-quarrel} can be restated as 
\begin{equation} \label{eqn:derived-rooted-quarrel}
	|\theta_{k,j}\tau - q_{\lambda_k}(a,b)| < \epsilon \pmod{2\pi},
	\ \ \mbox{ for $j=1,\ldots,m$.}
\end{equation}
for each $k=0,\dots,n$. Recall that from (\ref{eq:quarrel}), the circulant quarrels are given by 
\begin{equation}\label{eq:quarrel}
	q_{\lambda_k}(a,b) = 
	\frac{2\pi(b-a)k}{n}.
\end{equation}

As $C$ is a circulant with universal perfect state transfer, then by \cite[Theorem 22]{cfghst14}, we know that the eigenvalues of $C$ take the form
\[
\lambda_k = \beta + \gamma(kh+c_kn)
\]
for some $\beta,\gamma\in\RR$, $c_k\in\ZZ$, and $h\in\ZZ$ with $h$ relatively prime to $n$.  By making a diagonal shift, we may assume without loss of generality that $\beta=0$, and by adjusting the time parameter $t$ in $e^{-itC}$, we may assume  that $\gamma=1$ (see the discussion after Theorem 22 in \cite{cfghst14}).  Thus, without loss of generality, we may assume that
\begin{equation}\label{eq:evals}
\lambda_k = kh+c_kn\end{equation} where $\gcd(h,n)=1$ and $c_k\in\ZZ$.

\bigskip
To apply Kronecker's theorem, we assume that
\begin{equation} \label{eqn:kronecker}
	\sum_{k=0}^{n-1} \sum_{j=1}^{m} \ell_{k,j} \theta_{k,j} = 0
\end{equation}
where $\{\theta_{k,j} : j=1,\ldots,m\}$ are the roots of $R_{\lambda_k}(t)$. 
Then, we must show that
\[
	 \sum_{k=0}^{n-1} q_{\lambda_k}(a,b) \times \left(\sum_{j=1}^{m} \ell_{k,j} \right) 
	\equiv 0\pmod{2\pi}.
\]
From (\ref{eq:quarrel}), this becomes
\[
	 \sum_{k=0}^{n-1} \frac{2\pi(b-a)k}{n} \times \left(\sum_{j=1}^{m} \ell_{k,j} \right) 
	\equiv 0\pmod{2\pi}.
\]
Equivalently, it suffices to show
\begin{equation} \label{eqn:n-divisible}
	\sum_{k=0}^{n-1} k \left(\sum_{j=1}^{m} \ell_{k,j} \right) \equiv 0\pmod{n}.
\end{equation}

Let $F = \QQ(\alpha)$ be a field extension of the rationals.
Recall the the $\theta_{k,j}$ for $j=1,\dots,m$ are roots of the polynomials $R_{\lambda_k}(t)$ where we defined
\[
R_{\lambda_k}(t) = \phi(P_m^\alpha,t) - \lambda_k\phi(P_{m-1}^\alpha,t).
\]
Let $p_m(t)$ be the characteristic polynomial of the path $P_m$ on $m$ vertices.  Then Laplace expansion of the determinant yields that $\phi(P_m^\alpha,t) = tp_{m-1}(t) - \alpha^2p_{m-2}(t)$.
Thus, 
\begin{align}
	R_{\lambda_k}(t) &= t(p_{m-1}(t) -\lambda_k p_{m-2}(t))  - \alpha^2 (p_{m-2}(t) -\lambda_k p_{m-3}(t)).
\end{align}
Recall also from (\ref{eq:evals}) that we are assuming that $\lambda_k = kh+c_kn$ where $h,c_k\in\ZZ$, $k=0,\dots,n-1$.  In particular, each $\lambda_k\in\ZZ$.  We will show that these polynomials are irreducible.

\begin{lemma}
	\label{lemma:transcendental_q_implies_irreducibility}
	Let $p$ and $q$ be polynomials in $\QQ$ and $\alpha$ be transcendental over $\QQ$. Suppose that $(p, q) = 1$. Then $p + \alpha^2q$ is irreducible over $F=\QQ(\alpha)$. 
\end{lemma}
\begin{proof}
	Let $p$ and $q$ be polynomials in $\QQ$ and $\alpha$ be transcendental over $\QQ$. Suppose that $(p, q) = 1$. 
	
	

    Since $(p,q)=1$, $p$ and $q$ do not share a common factor, so any factorization is in the form of $p+\alpha^2q=(a+\alpha p)(c+\alpha d)$ where $a,b,c,d$ are polynomials in $\QQ$.  Then $ac=p$, $ad+bc=0$, and  $bd=q$.
	
	Suppose that $a$ is not a scalar multiple of $c$. Then $a$ and $b$ share a factor since $ad = -bc$. Thus, $ac$ and $bd$ share a factor. However, this can't happen since $ac = p$ and $bd = q$ and we assumed that $(p, q) = 1$. 
	
	Suppose that $a = kc$. Then $ad = kcd = - bc$ or $b = -kd$. So then we have that $(a + \alpha b)(c + \alpha d) = k(c^2 + \alpha^2d^2) = p + \alpha^2q$. But then $p$ and $q$ share a factor of $k$, a contradiction.
\end{proof}

\begin{lemma}
The polynomials $R_\lambda(t) \in F[x]$ are irreducible over $F$ for $\lambda$ ranging through the eigenvalues of $C$.

\begin{proof}

Let $R_\lambda(t) = t Q^\lambda_{m-1}(t) - \alpha^2 Q^\lambda_{m-2}(t)$, where  and $Q^\lambda_{\ell}(t) = p_\ell(t)-\lambda p_{\ell-1}(t)$. 
Note $Q^\lambda_{m-1}(t)$ and $Q^\lambda_{m-2}(t)$ are coprime (see Theorem \ref{thm:interlacing} of Appendix \ref{sec:app_orthpolyn}), and indeed $tQ^\lambda_{m-1}(t)$ and $Q^\lambda_{m-2}(t)$ are coprime (as the latter has only positive roots).  Thus by Lemma \ref{lemma:transcendental_q_implies_irreducibility}, each $R_\lambda$ is irreducible.

\end{proof}
\end{lemma}

The remainder of the proof will follow the machinery of \cite{eisenberg2019pretty}, making use of the field trace of a field extension.   See Appendix \ref{sec:galois} for details around Galois theory and the field trace.
Let $F_{k}$ be the splitting fields of $R_{\lambda_k}(t)$ over $F$, $k=0,\dots,n-1$.
These are Galois extensions since they are finite and normal (because they are splitting fields of the 
corresponding polynomials) and they are separable (since every irreducible polynomial is separable 
over fields of characteristic zero). 
Let $M$ be the smallest field extension of $F$ which contains all the  $F_{k}$, $k=1,\dots,n-1$.
This composite extension $M$ is also a Galois extension of $F$ (see \cite{df}, page 592).

\bigskip
\par\noindent{\em Notation}:
Given a polynomial $p(t) = a_d t^d + a_{d-1} t^{d-1} + \ldots + a_0$, 
we denote the coefficient of $p(t)$ corresponding to $t^j$ as $a_j = [t^j]p(t)$. We will focus primarily on the degree $d-1$ coefficient $[t^{d-1}]p(t)$ (the ``trace" of the polynomial).  

We apply the field trace to \eqref{eqn:kronecker} to get
\begin{eqnarray}
0 & = & \Tr_{M/F}\left( \sum_{k=0}^{n-1} \sum_{j=1}^{m} \ell_{k,j} \theta_{k,j} \right) \\
	& = &  \sum_{k=0}^{n-1} \left( \sum_{j=1}^{m} \ell_{k,j} \Tr_{M/F}(\theta_{k,j}) \right) \\
	& = & \sum_{k=0}^{n-1} [M:F_k] \sum_{j=1}^{m} \ell_{k,j} \Tr_{F_k/F}(\theta_{k,j}).
\end{eqnarray}
By the properties of the field trace, for each root  $\theta_{k,j}$ of $R_{\lambda_k}(t)$, $k=0,\ldots,n-1$, we have
\begin{equation}
	\Tr_{F_k/F}(\theta_{k,j}) 
	= \frac{[F_k:F]}{\deg(R_{\lambda_k})}[t^{m-1}]R_{\lambda_k}(t)
	= \frac{[F_k:F]}{m} (- \lambda_k)
\end{equation}
because the trace of the polynomials $R_{\lambda_k}$ are given by $[t^{m-1}]R_{\lambda_k}(t) = -\lambda_k$. 

Applying this to the preceding equality and using (\ref{eq:evals}), we obtain
\begin{align*}
0 &= \frac{[M:F]}{m} \left( \sum_{k=0}^{n-1} (-\lambda_k )\sum_{j=1}^{m} \ell_{k,j}
	\right)\\
    &= -\frac{[M:F]}{m}\left( \sum_{k=0}^{n-1} (kh+c_kn )\sum_{j=1}^{m} \ell_{k,j}.
	\right)
\end{align*}
Since $h$ is relatively prime to $n$, we conclude that,
\begin{equation}
	\sum_{k=0}^{n-1} k \left(\sum_{j=1}^{m} \ell_{k,j}\right) \equiv 0\pmod{n}
\end{equation}
which proves \eqref{eqn:n-divisible}, and the proof is complete.
\end{proof}
\end{theorem}

As an immediate corollary, we have a result for paths of even length (see Figure \ref{fig:even_path}).

\begin{corollary}\label{cor:even_path}
    For any path with an even number of vertices, if we place a transcendental edge weight on the edges incident to the endpoints, then there is PGST between the end vertices.
    \begin{proof}
        This follows immediately from the previous theorem taking the base graph $X$ to be the path on 2 vertices, which is a circulant with perfect state transfer (and hence universal state transfer since there are only two vertices).
    \end{proof}
\end{corollary}

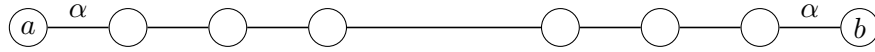
\begin{figure}[h]
    \centering

\begin{tikzpicture}[scale=1.1]
  \node[vertex] (A) at (0:1.4) {};
  \node[vertex] (B) at (180:1.4) {};
 
  \node[vertex] (A1) at (0:2.6) {};
  \node[vertex] (A2) at (0:3.8) {};
  \node[vertex, label=center:$b$] (A3) at (0:5.0) {};

  \node[vertex] (B1) at (180:2.6) {};
  \node[vertex] (B2) at (180:3.8) {};
  \node[vertex, label=center:$a$] (B3) at (180:5.0) {};

  \draw[undir] (A) -- (A1);
  \draw[undir] (A1) -- (A2);
  \draw[undir] (A2) -- node[midway, above] {$\alpha$} (A3);

  \draw[undir] (B) -- (B1);
  \draw[undir] (B1) -- (B2);
  \draw[undir] (B2) -- node[midway, above] {$\alpha$} (B3);

 \draw[undir] (A) -- (B);
\end{tikzpicture}

    \caption{Illustration of Corollary \ref{cor:even_path}. For transcendental $\alpha$, this will exhibit PGST between nodes $a$ and $b$.}
    \label{fig:even_path}
\end{figure}

\section{Conclusion and open questions}\label{sec:concl}

We end with some discussion of open questions.  In light of Corollary \ref{cor:even_path}, it is natural to ask if there is PGST in the same situation when the path has an odd number of vertices.

\begin{question}
    If we add a transcendental edge weight to the end edges of a path with an odd number of vertices, do we necessarily have PGST?
\end{question}

This does not follow from Theorem \ref{thm:main} in the same way that Corollary \ref{cor:even_path} did since an odd length path does not have the structure of a rooted product where the base graph is a circulant.  Indeed, if one attempts to follow the same proof strategy for odd length paths, the polynomials obtained that are analogous to the $R_\lambda(t)$ end up having trace 0, so the Galois theory argument used in attempting to use Kronecker's Theorem (Theorem \ref{thm:kronecker}) will not be informative. Some new tools will need to be developed to tackle this case.

\section*{Acknowledgments}

A.B. and C.T. were supported by NSF ExpandQISE grant 2427020.

	\bibliographystyle{plain}
	\bibliography{refs}

\appendix
    
\section{Galois Theory}\label{sec:galois}

We summarized some relevant facts from Galois theory. 
Let $E$ be a Galois extension of a field $F$ and let $\Gal(E/F)$ denote the Galois group of this extension.
The {\em trace} of an element $\alpha \in E$ is defined as
\begin{equation}
	\Tr_{E/F}(\alpha) = \sum_{\sigma \in \Gal(E/F)} \sigma(\alpha).
\end{equation}
This trace map has the following useful properties (see \cite{df,ft}, for example).

\begin{theorem} \label{thm:galois} 
Let $E/F$ be a Galois extension. The following properties hold:
\begin{enumerate}[(i)]
\item For all $\alpha \in E$, $\Tr_{E/F}(\alpha) \in F$.

\item For all $a \in F$, $\Tr_{E/F}(a) = [E:F]a$.

\item (linearity) For all $\alpha,\beta \in E$, 
	$\Tr_{E/F}(\alpha + \beta) = \Tr_{E/F}(\alpha) + \Tr_{E/F}(\beta)$.

\item (composability) If $F \subset K \subset E$ are extension fields, 
	then $\Tr_{E/F}(\alpha) = \Tr_{E/K}(\Tr_{K/F}(\alpha))$.

\item Let the minimal polynomial $m(x)$ of $\alpha \in E$ be of degree $d$.
	Then, $\Tr_{E/F}(\alpha) = -\frac{[E:F]}{d} [x^{d-1}]m(x)$.
\end{enumerate}

\begin{proof}
We sketch the proofs for completeness.
\begin{enumerate}[(i)]
\item Because $\sum_{\sigma \in \Gal(E/F)} \sigma(\alpha)$ is invariant under the action of the Galois group $\Gal(E/F)$,
	it must belong to the fixed field of $\Gal(E/F)$, which is $F$.

\item This is because $|\Gal(E/F)| = [E:F]$ and each element of $\Gal(E/F)$ fixes $F$.

\item Each element of the Galois group $\Gal(E/F)$ is a field $F$-automorphism.

\item Each element $\sigma$ of $\Gal(E/F)$ can be viewed as a composition of an element $\rho$ of $\Gal(E/K)$ 
	with an element $\tau$ of $\Gal(K/F)$ (see \cite{ft}, page 15). Therefore, $\Tr_{E/F}(\alpha)$ is given by
	\begin{equation}
	\sum_{\sigma \in \Gal(E/F)} \sigma(\alpha) 
		= \sum_{\substack{\rho \in \Gal(E/K)\\ \tau \in \Gal(K/F)}} \tau(\rho(\alpha)) 
		= \sum_{\tau \in \Gal(K/F)} \tau\left(\sum_{\rho \in \Gal(E/K)} \rho(\alpha) \right) 
	\end{equation}
	which yields $\Tr_{K/F}(\Tr_{E/K}(\alpha))$. 

\item Let $K = F(\alpha)$. Because $K$ is the splitting field of $m(x)$, we have
	$m(x) = \prod_{\sigma \in \Gal(K/F)} (x - \sigma(\alpha))$.
	So, $-\Tr_{K/F}(\alpha)$ is the second leading coefficient of $m(x)$.
	Now we apply (iv) and (ii) to obtain 
	\begin{equation}
	\Tr_{E/F}(\alpha) = \Tr_{K/F}(\Tr_{E/K}(\alpha)) = [E:K]\Tr_{K/F}(\alpha)
	\end{equation}
	which yields the claim because $[E:F] = [E:K][K:F]$ and $[K:F] = d$.

\end{enumerate}
\end{proof}
\end{theorem}


\section{Orthogonal Polynomials}\label{sec:app_orthpolyn}

We describe relevant facts about orthogonal polynomials (see Chihara \cite{chihara}). 
Let $(\mu_n)_{n=0}^{\infty}$ be a sequence of numbers and let $\EE$ be a real-valued
function defined on the vector space of all polynomials by
$\EE[x^n] = \mu_n$, for $n=0,1,\ldots$, and
$\EE[a_1 p_1(x) + a_2 p_2(x)] = a_1\EE[p_1(x)] + a_2\EE[p_2(x)]$ for all real numbers
	$a_1$ and $a_2$, and all polynomials $p_1(x)$ and $p_2(x)$.
Then, $\EE$ is called a {\em moment functional} defined by $(\mu_n)_{n=0}^{\infty}$. 
A moment functional $\EE$ is called {\em positive definite} on $E \subset (-\infty,\infty)$ if
$\EE[f] > 0$ for all real polynomial $f$ which is non-negative on $E$ (and does not vanish
identically on $E$). The set $E$ is called the suppor of $\EE$.

A sequence of polynomials $(Q_n(x))_{n=0}^{\infty}$ is called an {\em orthogonal polynomial sequence} 
with respect to a moment functional $\EE$ if for all non-negative integers $m$ and $n$ we have
\begin{enumerate}[(i)]
\item $Q_n(x)$ is a polynomial of degree $n$;
\item $\EE[Q_m(x)Q_n(x)] = 0$ for $m \neq n$;
\item $\EE[Q_n^2(x)] \neq 0$. (It is more convenient to assume strictly positive)
\end{enumerate}
We typically assume that the leading coefficient of $Q_n(x)$ is positive (or one in the monic case).

\begin{theorem} \label{thm:favard}
(Favard, see \cite{chihara})
Let $(a_n)_{n=1}^{\infty}$ and $(b_n)_{n=1}^{\infty}$ be arbitrary sequences of numbers.
Let $(Q_n(x))_{n=0}^{\infty}$ be a sequence of polynomials defined by the recurrence
\begin{equation} \label{eqn:three-term}
Q_n(x) = (x - a_n)Q_{n-1}(x) - b_n Q_{n-2}(x), \ \ \mbox{$n=1,2,\ldots$}
\end{equation}
with $Q_0(x) = 1$ and $Q_{-1}(x) = 0$.
Then, there is a unique moment functional $\EE$ so that
$\EE[1] = 1$ and $\EE[Q_m(x)Q_n(x)] = 0$ for $m \neq n$ with $m,n = 0,1,\ldots$.
Moreover, $\EE$ is positive definite if and only if $a_n$ is real and $b_n > 0$ (for $n \ge 1$).
\end{theorem}

\begin{theorem} \label{thm:christoffel-darboux}
(Christoffel-Darboux, see \cite{chihara})
Let $(Q_n(x))$ be a sequence of polynomials which satisfy \eqref{eqn:three-term} with $b_n \neq 0$. Then
\begin{equation} \label{eqn:orthogonality}
\sum_{k=0}^{m-1} \frac{Q_k(x)Q_k(y)}{b_1 b_2 \ldots b_{k+1}} 
	=
	\frac{1}{b_1 b_2 \ldots b_n} 
	\frac{Q_{m}(x)Q_{m-1}(y) - Q_{m-1}(x)Q_{m}(y)}{x-y}.
\end{equation}
\end{theorem}

\par\noindent
A corollary of Theorem \ref{thm:christoffel-darboux} 
when $x$ and $y$ are two distinct roots of $Q_m$ where $b_n = 1$, for all $n$, is
\begin{equation}
\sum_{k=0}^{m-1} Q_k(x)Q_k(y) = 0.
\end{equation}

\begin{theorem}
Let $E$ be the support of $\EE$. The zeros of $Q_n(x)$ are all real, simple
and are located in the interior of $E$.
\end{theorem}

\par\noindent
We denote the zeros of $Q_n(x)$ by $x_{n,i}$ where $i=1,\ldots,n$ and
\begin{equation}
x_{n,1} < x_{n,2} < \ldots < x_{n,n}.
\end{equation}
\ignore{
Since the leading coefficient of $Q_n(x)$ is positive, we have
$Q_n(x) > 0$ for $x > x_{n,n}$ and $\sgn(Q_n(x)) = (-1)^n$ for $x < x_{n,1}$,
where $\sgn(z) = 1$ if $z > 0$, $0$ if $z = 0$ and $-1$ if $z < 0$.
}

\begin{theorem}\label{thm:interlacing} (Strict Interlacing)
The zeros of $Q_n(x)$ and $Q_{n+1}(x)$ mutually separate each other; that is,
$x_{n+1,i} < x_{n,i} < x_{n+1,i+1}$, for $i=1,2,\ldots,n$.
\end{theorem}

\end{document}